\documentclass{article}
\usepackage{amsfonts}
\usepackage{graphicx}
\usepackage{amsmath,amssymb,amsthm}
\usepackage{authblk}
\usepackage{url}

\newcommand{\ST}{\ensuremath{\mathsf{ST}}}
\newcommand{\SLT}{\ensuremath{\mathsf{SLT}}}

\newcommand{\SL}{\mathtt{suffixLink}}

\newcommand{\BWT}{\ensuremath{\mathsf{BWT}}}
\newcommand\SP[1]{\mathtt{sp}(#1)}
\newcommand\EP[1]{\mathtt{ep}(#1)} 
\newcommand\INTERVAL[1]{\mathtt{range}(#1)}

\newcommand{\repr}{\ensuremath{\mathtt{repr}}}
\newenvironment{proofSketch}{\noindent {\it Proof sketch.}}{$\Box$\vskip1ex}

\newtheorem{lemma}{Lemma} 
\newtheorem{theorem}{Theorem}

\begin{document} 

\title{Space-efficient detection of unusual words}
\author[1,2]{Djamal Belazzougui}
\author[3]{Fabio Cunial}
\affil[1]{Department of Computer Science, University of Helsinki, Finland.\thanks{This work was partially supported by Academy of Finland under grant 284598 (Center of Excellence in Cancer Genetics Research).}}
\affil[2]{Helsinki Institute for Information Technology, Finland.}
\affil[3]{Max Planck Institute of Molecular Cell Biology and Genetics, Dresden, Germany.}
\maketitle

\begin{abstract}
Detecting all the strings that occur in a text more frequently or less frequently than expected according to an IID or a Markov model is a basic problem in string mining, yet current algorithms are based on data structures that are either space-inefficient or incur large slowdowns, and current implementations cannot scale to genomes or metagenomes in practice. In this paper we engineer an algorithm based on the suffix tree of a string to use just a small data structure built on the Burrows-Wheeler transform, and a stack of $O(\sigma^2\log^2 n)$ bits, where $n$ is the length of the string and $\sigma$ is the size of the alphabet. The size of the stack is $o(n)$ except for very large values of $\sigma$. We further improve the algorithm by removing its time dependency on $\sigma$, by reporting only a subset of the maximal repeats and of the minimal rare words of the string, and by detecting and scoring candidate under-represented strings that \emph{do not occur} in the string. Our algorithms are practical and work directly on the BWT, thus they can be immediately applied to a number of existing datasets that are available in this form, returning this string mining problem to a manageable scale.
\end{abstract}

\section{Introduction} \label{sec:intro}

Detecting all the patterns of a string whose number of occurrences matches some notion of statistical surprise is a fundamental requirement of the post-genome era, in which textual datasets grow faster than the ability to understand them, and in which over- and under-representation with respect to a statistical model is often an indicator of structure or function. The sheer volume of the available datasets makes even simple models of patterns and simple measures of statistical surprise useful in practice, if their detection scales to extremely long strings in reasonable time and space. In this paper we focus on the simplest possible model of a pattern -- a string $W$, of any length, that occurs without mismatches $f_{T}(W)$ times in a text $T$ of length $n$ -- and we consider measures of statistical surprise that score $W$ according to the expected number $\mathbb{E}[f_{T}(W)]$ and to the variance $\mathbb{V}[f_{T}(W)]$ of the number of its occurrences in a random text of length $|T|$ (see e.g. \cite{apostolico2003monotony} and references therein). We assume that the random source is a given Markov chain, and for concreteness we set its order to zero, since this simple case already captures the computational structure of the problem \cite{apostolico2000efficient}. Moreover, we focus on computing $\mathbb{V}[f_{T}(W)]$, since computing expectations with respect to a Markov chain of order zero is trivial (see e.g. \cite{apostolico1997annotated}).

$\mathbb{V}[f_{T}(W)]$ enjoys the remarkable property that its computation can be carried out by iterating over all the proper \emph{borders} of $W$, i.e. over all the nonempty substrings of $W$ shorter than $W$ that are at the same time prefix and suffix of $W$: see \cite{apostolico1997annotated} for a detailed derivation. To make the paper self-contained, here we just recall that $\mathbb{V}[f_{T}(W)]$ can be computed in constant time from the functions $\phi(W)$ and $\gamma(W)$ defined below:
\begin{eqnarray*}
\phi(W) & = & \sum_{b \in \mathcal{B}(W)}(n-2|W|+b+1) \cdot \pi(W[b..|W|-1]) \\
\gamma(W) & = & \sum_{b \in \mathcal{B}(W)} \pi(W[b..|W|-1]) \\
\end{eqnarray*}
where strings are indexed from zero, $\pi(W)=\prod_{i=0}^{|W|-1}\mathbb{P}[W_i]$, $W_i$ is the $i$th character of string $W$, $\mathbb{P}[c]$ is the probability of character $c$ according to the given zero-order Markov chain, and $\mathcal{B}(W)$ is the set of all border lengths of $W$. Removing the components of $\mathbb{V}[f_{T}(W)]$ that depend on borders can cause large relative errors in practice \cite{apostolico2000efficient}, so we focus on the exact computation of $\mathbb{V}[f_{T}(W)]$. It is well known that borders have a recursive structure, in the sense that the set of borders of $W$ consists of the longest border $V$ of $W$, and of all the borders of $V$. This observation enables one to map the computation of $\mathbb{V}[f_{T}(W)]$ for a given $W$ onto the Morris-Pratt algorithm \cite{morris1970linear}, thus achieving time $O(|W|)$ in the worst case \cite{apostolico1997annotated}. Specifically:
\begin{eqnarray*}
\phi(W) & = & \delta(W) \cdot \big( \phi(B) -2(|W|-|B|)\gamma(B) +n-2|W|+|B|+1 \big)\\
\gamma(W) & = & \delta(W) \cdot \big( 1+\gamma(B) \big)
\end{eqnarray*}
where $B$ is the longest border of $W$, and $\delta(W) = \pi(W[|B|..|W|-1])$. In practice, however, we are interested in extracting from a string $T$ \emph{all its substrings $W$}, \emph{of any length}, such that a user-specified measure of surprise computed on $\mathbb{E}[f_{T}(W)]$ and $\mathbb{V}[f_{T}(W)]$ is, say, greater than a threshold. Even though computing $\mathbb{V}[f_{T}(W)]$ takes $O(|W|)$ time for a given $W$, enumerating and scoring in this way all substrings of a text $T$ of length $n$ takes $O(n^2)$ time.

Luckily, a number of statistical scores $z(W)$ enjoy the additional property that $z(XWY) \geq z(W)$ if $f_{T}(XWY) = f_{T}(W)$, where $X$ and $Y$ are strings \cite{apostolico2003monotony}. Consider then a set $\mathcal{A}$ of substrings of $T$ such that all substrings in the set have the same number of occurrences, and consider the partial order $\preceq$ on $\mathcal{A}$ such that $V \preceq W$ iff $W=XVY$ for (possibly empty) strings $X$ and $Y$. If we display to the user just the maximal elements of $\mathcal{A}$ with respect to $\preceq$, we guarantee that every over-represented string $W \in \mathcal{A}$ that we do not output is a substring of a string $XWY \in \mathcal{A}$ in the output which has at least the same score. Symmetrically, if we display just the minimal elements of $\mathcal{A}$ with respect to $\preceq$, we guarantee that every under-represented string $XWY \in \mathcal{A}$ that we do not output is a superstring of a string $W \in \mathcal{A}$ in the output which has at most the same score. A possible choice for $\mathcal{A}$ is the set of all substrings that start at exactly the same positions in $T$: this class has a unique maximal element, which corresponds to a node of the suffix tree of $T$, and a unique minimal element, which corresponds to the right extension by a single character of a node of the suffix tree of $T$. Since the number of all such classes is $O(|T|)$, and since all such classes are connected to one another by a trie, known as the \emph{suffix-link tree}, it is possible to devise an algorithm that computes $\mathbb{V}[f_{T}(W)]$ for all the maximal and minimal elements of all such classes, in a total amount of time that grows linearly in $|T|$ \cite{apostolico2000efficient}.

In this paper we engineer the algorithm described in \cite{apostolico2000efficient} to use as its substrate the Burrows-Wheeler transform of $T$, rather than space-inefficient data structures like the (truncated) suffix tree of $T$, or space-efficient simulations of the suffix tree with $O(\log^{\varepsilon}{n})$ slowdown, like compressed suffix trees (see e.g. \cite{gog2011compressed} and references therein). We also observe that the time complexity of the algorithm described in \cite{apostolico2000efficient} depends on the cardinality of the alphabet, and we remove this dependency. Moreover, we adapt the algorithm to work on the smallest possible set of equivalence classes, thus reducing time and output size in practice. Assuming an alphabet of size $\sigma \in o(\sqrt{n}/\log{n})$, we can thus perform all the computations described in \cite{apostolico2000efficient} in $O(n)$ time and in $o(n+\lambda\sqrt{n})$ bits of space, given the BWT of $T$ and few additional data structures \cite{BNV13}, where $\lambda$ is the length of a longest repeat of $T$. For statistical reasons, the maximum length of a string to be reported is often $O(\log_{\sigma}{n})$, thus our algorithm uses effectively $o(n)$ bits of space in addition to the input. Concatenating this setup to the BWT construction algorithm described in \cite{Belazzougui14}, we can discover all the over- and under-represented substrings of $T$, \emph{directly from $T$ itself}, in randomized $O(n)$ time and in $O(n\log{\sigma})$ bits of space in addition to $T$ itself. Finally, we extend the algorithm in \cite{apostolico2000efficient} to consider potentially under-represented strings that do not occur in $T$, thus providing for the first time a way to score and rank the minimal absent words of $T$.

%

\section{Preliminaries}

\subsection{Strings} \label{sec:strings}

Let $\Sigma=[1..\sigma]$ be an integer alphabet, let $\#=0$ be a separator not in $\Sigma$, let $T=[1..\sigma]^{n-1}\#$ be a string, and let $\varepsilon$ be the empty string. For reasons that will become clear in Section \ref{sec:enumerator}, we assume $\sigma \in o(\sqrt{n}/\log{n})$ throughout the paper. We denote by $f_{T}(W)$ the number of (possibly overlapping) occurrences of a string $W$ in the circular version of $T$. A \emph{repeat} $W$ is a string that satisfies $f_{T}(W)>1$. We denote by $\Sigma^{\ell}_{T}(W)$ the set of characters $\{a \in [0..\sigma] : f_{T}(aW)>0\}$ and by $\Sigma^{r}_{T}(W)$ the set of characters $\{b \in [0..\sigma] : f_{T}(Wb)>0\}$. A repeat $W$ is \emph{right-maximal} (respectively, \emph{left-maximal}) iff $|\Sigma^{r}_{T}(W)|>1$ (respectively, iff $|\Sigma^{\ell}_{T}(W)|>1$). It is well known that $T$ can have at most $n-1$ right-maximal substrings and at most $n-1$ left-maximal substrings. A \emph{maximal repeat} of $T$ is a repeat that is both left- and right-maximal. Clearly a maximal repeat $W$ of $T$ satisfies $f_{T}(aW)<f_{T}(W)$ and $f_{T}(Wb)<f_{T}(W)$ for any characters $a$ and $b$ in $\Sigma$. A repeat is \emph{supermaximal} if it is not a proper substring of any other repeat. A \emph{minimal rare word} of $T$ is a string $W$ that satisfies $f_{T}(W)<f_{T}(V)$ for every proper substring $V$ of $W$. Clearly $W$ must have the form $aXb$, where $a$ and $b$ are characters and $X$ is a maximal repeat of $T$. If $f_{T}(W)>0$, then $aX$ is a right-maximal substring of $T$ and $Xb$ is a left-maximal substring of $T$. If $f_{T}(W)=0$, then $aX$ (respectively, $Xb$) must occur in $T$, but it is not necessarily right-maximal (respectively, left-maximal). A minimal rare word of $T$ that does not occur in $T$ is called \emph{minimal absent word} (see e.g. \cite{chairungsee2012using,herold2008efficient} and references therein): the total number of such strings can be $\Theta(\sigma n)$ \cite{crochemore1998automata}. A minimal rare word of $T$ that occurs exactly once in $T$ is called \emph{minimal unique substring} (see e.g. \cite{ileri2015shortest} and references therein). It is clear that the total number of minimal rare words of $T$ that occur at least once in $T$ is $O(n)$.

String $V \neq \varepsilon$ is a \emph{proper border} of string $W$ if $W=VX$ and $W=YV$ for nonempty strings $X$ and $Y$. A string $W$ can have zero, one, or multiple proper borders: we denote by $bord(W)$ the length of the \emph{longest} border of $W$. Each border of $W$ is followed by a character when it occurs as a prefix, and it is preceded by a character when it occurs as a suffix: we use $a|W$ to denote the length of the longest border of $W$ that is preceded by character $a$ when it occurs as a suffix, and we use $W|a$ to denote the length of the longest border of $W$ that is followed by character $a$ when it occurs as a prefix. Clearly both $a|W$ and $W|a$ can be zero. We denote by $\mathcal{B}(W)$ the set of lengths of all borders of $W$, by $\mathcal{B}^{r}(W)$ the set of pairs $\{ (a,a|W) : a \in \sigma, a|W \neq 0 \}$, and by $\mathcal{B}^{\ell}(W)$ the set of pairs $\{ (a,W|a) : a \in \sigma, W|a \neq 0 \}$. It is well known that $\mathcal{B}(W) = \{ bord(W) \} \cup \mathcal{B}(V)$, where $V$ is the longest border of $W$: see e.g. \cite{Jewels_of_stringology} and references therein. In this paper we will also use $\mathtt{left}_W$ to denote an array of size $|\Sigma_{T}^{r}(W)|$, indexed by the characters in $\Sigma_{T}^{r}(W)$ in lexicographic order, such that $\mathtt{left}_{W}[c]=W|c$. Similarly, we will use $\mathtt{right}_W$ to denote an array of size $|\Sigma_{T}^{\ell}(W)|$, indexed by the characters in $\Sigma_{T}^{\ell}(W)$ in lexicographic order, such that $\mathtt{right}_{W}[c]=c|W$. Set $\mathcal{B}(W)$ determines all possible ways in which $W$ can overlap with itself: specifically, the maximum possible number of occurrences of $W$ in a string of length $n$ is $f^{*}(W,n) = \lceil (n-|W|+1)/period(W) \rceil$, where $period(W)=|W|-bord(W)$. When $f_{T}(W)$ needs to be compared to $f_{T'}(W)$, where $|T'| \neq |T|$, it is customary to divide $f_{T}(W)$ by $f^{*}(W,|T|)$. It is easy to see that the longest border of a random string of length $n$ generated by an IID source is expected to tend to a constant as $n$ tends to infinity.

For reasons of space we assume the reader to be familiar with the notion of \emph{suffix tree} $\ST_T$ of a string $T$, which we do not define here. 
We denote by $\ell(v)$ the string label of a node $v$ in a suffix tree. It is well known that a substring $W$ of $T$ is right-maximal iff $W=\ell(v)$ for some internal node $v$ of $\ST_T$. We assume the reader to be familiar with the notion of \emph{suffix link} connecting a node $v$ with $\ell(v)=aW$ for some $a \in [0..\sigma]$ to a node $w$ with $\ell(w)=W$: we say that $w=\SL(v)$ in this case. Here we just recall that suffix links and internal nodes of $\ST_T$ form a tree, called the \emph{suffix-link tree} of $T$ and denoted by $\SLT_T$, and that inverting the direction of all suffix links yields the so-called \emph{explicit Weiner links}. Given an internal node $v$ and a symbol $a \in [0..\sigma]$, it might happen that string $a\ell(v)$ does occur in $T$, but that it is not right-maximal, i.e. it is not the label of any internal node of $\ST_T$: all such left extensions of internal nodes that end in the middle of an edge are called \emph{implicit Weiner links}. An internal node $v$ of $\ST_T$ can have more than one outgoing Weiner link, and all such Weiner links have distinct labels: in this case, $\ell(v)$ is a maximal repeat. 
%
It is known that the number of suffix links (or, equivalently, of explicit Weiner links) is upper-bounded by $2n-2$, and that the number of implicit Weiner links can be upper-bounded by $2n-2$ as well.

If $V$ is a nonempty proper border of $W$, then $\Sigma_{T}^{\ell}(W) \subseteq \Sigma_{T}^{\ell}(V)$ and $\Sigma_{T}^{r}(W) \subseteq \Sigma_{T}^{r}(V)$. Thus, if $W$ is right-maximal (respectively, left-maximal) then $V$ is right-maximal (respectively, left-maximal); if $W$ is a maximal repeat, then $V$ is a maximal repeat; and if $W=aX$ where $a \in \Sigma$ and $X$ is a maximal repeat, then $V=aY$ where $Y$ is a maximal repeat\footnote{Thus, maximal repeats connected by longest border relationships form a tree rooted at the empty string: the path from the root to a maximal repeat lists all its borders, and the internal nodes of this tree cannot be supermaximal repeats. Similarly, longest border relationships and strings $aW$ (respectively, $Wa$) where $W$ is a maximal repeat, form a tree rooted at the empty string.}.

\subsection{Enumerating maximal repeats and minimal rare words}\label{sec:enumerator}

For reasons of space we assume the reader to be familiar with the notion and uses of the Burrows-Wheeler transform of $T$, including the $C$ array, the $\mathtt{rank}$ function, and backward searching. In this paper we use $\BWT_T$ to denote the BWT of $T$, we use $\INTERVAL{W} = [\SP{W}..\EP{W}]$ to denote the lexicographic interval of a string $W$ in a BWT that is implicit from the context, and we use $\Sigma_{i,j}$ to denote the set of distinct characters that occur inside interval $[i..j]$ of a string that is implicit from the context. We also denote by $\mathtt{rangeDistinct}(i,j)$ the function that returns the set of tuples $\{(c,\mathtt{rank}(c,p_c),\mathtt{rank}(c,q_c)) : c \in \Sigma_{i,j} \}$, \emph{in any order}, where $p_c$ and $q_c$ are the first and the last occurrence of character $c$ inside interval $[i..j]$, respectively. Here we focus on a specific application of $\BWT_T$: enumerating all the right-maximal substrings of $T$, or equivalently all the internal nodes of $\ST_T$. In particular, we use the algorithm described in \cite{Belazzougui14} (Section 4.1), which we sketch here for completeness.

Given a substring $W$ of $T$, let $b_1 < b_2 < \dots < b_k$ be the sorted sequence of all the distinct characters in $\Sigma^{r}_{T}(W)$, and let $a_1,a_2,\dots,a_h$ be the sequence of all the characters in $\Sigma^{\ell}_{T}(W)$, not necessarily sorted. Assume that we represent a substring $W$ of $T$ as a pair $\repr(W)=(\mathtt{chars}[1..k],\mathtt{first}[1..k+1])$, where $\mathtt{chars}[i]=b_i$, $\INTERVAL{Wb_i}=[\mathtt{first}[i]..\mathtt{first}[i+1]-1]$ for $i \in [1..k]$, and function $\INTERVAL$ refers to $\BWT_T$. Note that $\INTERVAL{W}=[\mathtt{first}[1]..\mathtt{first}[k+1]-1]$, since it coincides with the concatenation of the intervals of the right extensions of $W$ in lexicographic order. If $W$ is not right-maximal, array $\mathtt{chars}$ in $\repr(W)$ has length one. Given a data structure that supports $\mathtt{rangeDistinct}$ queries on $\BWT_T$, and given the $C$ array of $T$, there is an algorithm that converts $\repr(W)$ into the sequence $a_1,\dots,a_h$ and into the corresponding sequence $\repr(a_{1}W),\dots,\repr(a_{h}W)$, in $O(de)$ time and $O(\sigma^{2}\log{n})$ bits of space in addition to the input and the output \cite{Belazzougui14}, where $d$ is the time taken by the $\mathtt{rangeDistinct}$ operation per element in its output, and $e$ is the number of distinct strings $a_{i}Wb_{j}$ that occur in the circular version of $T$, where $i \in [1..h]$ and $j \in [1..k]$. We encapsulate this algorithm into a function that we call $\mathtt{extendLeft}$.

If $a_i W$ is right-maximal, i.e. if array $\mathtt{chars}$ in $\repr(a_{i}W)$ has length greater than one, we push pair $(\repr(a_i W),|W|+1)$ onto a stack $S$. In the next iteration we pop the representation of a string from the stack and we repeat the process, until the stack becomes empty. This process is equivalent to following all the explicit Weiner links from the node $v$ of $\ST_T$ with $\ell(v)=W$, not necessarily in lexicographic order. Thus, running the algorithm from a stack initialized with $\repr(\varepsilon)$ is equivalent to performing a depth-first preorder traversal of the suffix-link tree of $T$ 
(but with an arbitrary exploration order on the children of each node), which guarantees to enumerate all the right-maximal substrings of $T$. 
Every operation performed by the algorithm can be charged to a distinct node or Weiner link of $\ST_T$, thus the algorithm runs in $O(nd)$ time. We keep the depth of the stack to $O(\log{n})$ rather than to $O(n)$ by using the folklore trick of pushing at every iteration the pair $(\repr(a_i W),|a_i W|)$ with largest $\INTERVAL{a_i W}$ first (see e.g. \cite{hoare1962quicksort}). Every suffix-link tree level in the stack contains at most $\sigma$ pairs, and each pair takes at most $\sigma\log{n}$ bits of space, thus the total space used by the stack is $O(\sigma^2 \log^{2}{n})$ bits. The following theorem follows from our assumption that $\sigma \in o(\sqrt{n}/\log{n})$:

\begin{sloppypar}
\begin{theorem}[\cite{Belazzougui14}]\label{thm:enumerator}
Let $T \in [1..\sigma]^{n-1}\#$ be a string. Given a data structure that supports $\mathtt{rangeDistinct}$ queries on $\BWT_T$, we can enumerate all the right-maximal substrings $W$ of $T$, and for each of them we can return $|W|$, $\repr(W)$, the sequence $a_1,a_2,\dots,a_h$ of all characters in $\Sigma_{T}^{\ell}(W)$ (not necessarily sorted), and the sequence $\repr(a_{1}W),\dots,\repr(a_{h}W)$, in $O(nd)$ time and in $O(\sigma^2\log^2 n)=o(n)$ bits of space in addition to the input and the output, where $d$ is the time taken by the $\mathtt{rangeDistinct}$ operation per element in its output.
\end{theorem}
\end{sloppypar}

Theorem \ref{thm:enumerator} does not specify the order in which the right-maximal substrings must be enumerated, nor the order in which the left extensions of a right-maximal substring must be returned. The algorithm we just described can be adapted to return all the maximal repeats of $T$, within the same bounds, by outputting a right-maximal string $W$ iff $|\mathtt{rangeDistinct}(\SP{W},\EP{W})|>1$. Computing the minimal rare words that occur in $T$ is also easy:

\begin{lemma}\label{lemma:minimalRareWords}
Let $T \in [1..\sigma]^{n-1}\#$ be a string. Given a data structure that supports $\mathtt{rangeDistinct}$ queries on $\BWT_T$, we can enumerate all the minimal rare words $W$ of $T$ that occur at least once in $T$, and for each of them we can return $|W|$ and $\INTERVAL{W}$, in $O(nd)$ time and in $O(\sigma^2\log^2 n)=o(n)$ bits of space in addition to the input and the output, where $d$ is the time taken by the $\mathtt{rangeDistinct}$ operation per element in its output.
\end{lemma}
\begin{proof}
We use a technique similar to the one described in \cite{belazzougui2015framework}. Specifically, we adapt Theorem \ref{thm:enumerator} to iterate over all maximal repeats of $T$, and we allocate a temporary array $\mathtt{freq}[0..\sigma]$, indexed by all characters in the alphabet. After having enumerated a maximal repeat $W$, we scan $\mathtt{repr}(W)$, we compute the number of occurrences of every right extension $Wb$ of $W$ using array $\mathtt{first}$, and we write $f_{T}(Wb)$ in $\mathtt{freq}[b]$. Then, for every $i \in [1..h]$, we check whether $\mathtt{repr}(a_{i}W)$ contains more than one character: if this is the case, then $f_{T}(a_{i}W)>f_{T}(a_{i}Wb)$ for every $b \in [0..\sigma]$. Thus, we scan $\mathtt{repr}(a_{i}W)$ and for every $b$ in its array $\mathtt{chars}$ we check whether $f_{T}(a_{i}Wb)<f_{T}(Wb)$, by accessing $\mathtt{freq}[b]$. If this is the case, then $a_{i}Wb$ is a minimal rare word, and its interval in $\BWT_T$ can be derived in constant time from the array $\mathtt{first}$ of $\mathtt{repr}(a_{i}W)$. At the end of this process, we reset array $\mathtt{freq}$ to its initial state by scanning $\mathtt{repr}(W)$ again.
{$\Box$\vskip1ex}\end{proof}

Minimal rare words that do not occur in $T$ can be enumerated using a slight variation of Lemma \ref{lemma:minimalRareWords}, as described in \cite{belazzougui2015framework}:

\begin{lemma}[\cite{belazzougui2015framework}]\label{lemma:minimalAbsentWords}
Let $T \in [1..\sigma]^{n-1}\#$ be a string. Given a data structure that supports $\mathtt{rangeDistinct}$ queries on $\BWT_T$, we can enumerate all the minimal rare words $aWb$ of $T$ that do not occur in $T$, where $a$ and $b$ are characters and $W \in [1..\sigma]^*$, and for each of them we can return $a$, $b$, $|W|$, and $\INTERVAL{W}$, in $O(nd+\mathtt{occ})$ time and in $O(\sigma^2\log^2 n)=o(n)$ bits of space in addition to the input and the output, where $d$ is the time taken by the $\mathtt{rangeDistinct}$ operation per element in its output, and $\mathtt{occ}$ is the output size.
\end{lemma}

For reasons of space, we assume throughout the paper that $d$ is the time per element in the output of a $\mathtt{rangeDistinct}$ data structure that is implicit from the context.

\section{Computing the border of all right-maximal substrings} \label{appendix:borders}

As mentioned in the introduction, computing the exact variance of the number of occurrences of a string $W$ in a random text of length $n$ can be mapped to the computation of the longest border of all suffixes of $W$ \cite{apostolico1997annotated}, and thus takes $O(|W|)$ time using the Morris-Pratt algorithm \cite{morris1970linear}. To compute the longest border of \emph{all} right-maximal substrings of a text $T$, as well as of \emph{all} substrings $Wb$ of $T$ such that $b \in [0..\sigma]$ and $W$ is right-maximal, in overall linear time on $|T|$, we need the following algorithm described in \cite{apostolico2000efficient}, which we sketch here for completeness:

\begin{theorem}[\cite{apostolico2000efficient}] \label{theorem:apostolico}
Let $T \in [1..\sigma]^{n}$ be a string. There is an algorithm that computes $bord(W)$ for all right-maximal substrings $W$ of $T$, and for all substrings $W=Vb$ of $T$ where $b \in [1..\sigma]$ and $V$ is right-maximal, in $O(n)$ words of space. The running time of this algorithm is linear in $n$ and depends on $\sigma$.
\end{theorem}
\begin{proofSketch}
We build the suffix tree $\ST_T$ of $T$, and we assume that every node $v$ of $\ST_T$ stores sets $\Sigma_{T}^{\ell}(\ell(v))$ and $\Sigma_{T}^{r}(\ell(v))$ as lexicographically sorted lists. We perform a depth-first traversal of the \emph{suffix-link tree} of $T$, and we store in each node $v$ with label $\ell(v)=W$ the arrays $\mathtt{right}_{W}$ and $\mathtt{left}_{W}$ described in Section \ref{sec:strings}: recall that $\mathtt{right}_{W}$ (respectively, $\mathtt{left}_{W}$) is indexed by all characters $a \in \Sigma_{T}^{\ell}(W)$ (respectively, $b \in \Sigma_{T}^{r}(W)$), and it stores value $a|W$ at position $\mathtt{right}_{W}[a]$ (respectively, value $W|b$ at position $\mathtt{left}_{W}[b]$). Clearly, if $\mathtt{right}_{W}[a]>0$, then $bord(aW) = \mathtt{right}_{W}[a]+1$. If $\mathtt{right}_{W}[a]=0$, then $bord(aW)$ is either one (if $a$ matches the last character of $W$) or zero. Once we know $bord(aW)$, we compute array $\mathtt{right}_{aW}$ by exploiting the identity $\mathcal{B}(aW) = \{bord(aW)\} \cup \mathcal{B}(V)$, where $V$ is the longest border of $aW$ (see Section \ref{sec:strings}): specifically, for every character $c \in \Sigma_{T}^{\ell}(aW)$, we know that $c$ belongs also to $\Sigma_{T}^{\ell}(V)$, thus we set $\mathtt{right}_{aW}[c] = \mathtt{right}_{V}[c]$ if $c \neq d$, and we set $\mathtt{right}_{aW}[d] = bord(aW)$, where $d$ is the character that precedes the suffix of $aW$ of length $bord(aW)$. Since we know that every character $c \in \Sigma_{T}^{r}(aW)$ also belongs to $\Sigma_{T}^{r}(V)$, we can compute array $\mathtt{left}_{aW}$ from $\mathtt{left}_{V}$ in the same way.

Every cell of every array $\mathtt{right}_{W}$ can be charged to a (possibly implicit) Weiner link of $\ST_T$, and every cell of every array $\mathtt{left}_{W}$ can be charged to an edge of $\ST_T$, thus the algorithm uses $O(n)$ words of space overall in addition to $\ST_T$. However, copying $\mathtt{right}_{aW}[c]$ from $\mathtt{right}_{V}[c]$, where $V$ is the longest border of $aW$ (respectively, $\mathtt{left}_{aW}[c]$ from $\mathtt{left}_{V}[c]$) requires retrieving $c$ from the list of left extensions (respectively, right extensions) of $V$, or merging such list with the corresponding list of $aW$, which introduces a dependency on $\sigma$.
\end{proofSketch}

A dependency on $\sigma$ can be problematic in data mining applications, where the alphabet could be the result of a dense discretization of a continuous range (see e.g. \cite{keogh2002finding,lin2007experiencing} and references therein). To make Theorem \ref{theorem:apostolico} independent of $\sigma$, we start from generalizing the algorithm described in \cite{String_matching_algorithms_and_automata} to a trie, counting \emph{return arcs} that do not point to the root:
\begin{lemma} \label{lemma:returnArcs}
Let $T=(V,E,\sigma)$ be a trie on alphabet $[1..\sigma]$, let $\ell(v)=\ell(e_k) \cdot \ell(e_{k-1}) \cdot \cdots \cdot \ell(e_1)$ be the label of a node $v \in V$ that is reachable from the root with path $e_1, e_2, \dots, e_k$, where $e_i \in E$ for all $i \in [1..k]$, and let $a|v = w \in V : \ell(w)=a|\ell(v)$. The set of \emph{return arcs} $E'=\{(v,a|v) : v \in V, a \in [1..\sigma], a|v \neq \emptyset\}$ satisfies $|E'| \leq 2|V|$. 
\end{lemma}
\begin{proof}
We say that a return arc $(v,a|v) \in E'$ is of \emph{type 1} if there is an edge $e=(v,w) \in E$ with $\ell(e)=a$ , and we say that it is of \emph{type 2} otherwise. The total number of type-1 return arcs is at most $|V|$, thus we focus on type-2 return arcs. Let $A=\ell(v)$ and let $B=A[1..a|\ell(v)]$. We charge a type-2 return arc $(v,a|v)$ to the vertex $w$ that satisfies $A = B \cdot \ell(w)$. Assume that two distinct type-2 return arcs $(u_1,v_1)$, $(u_2,v_2)$ are charged to the same vertex $w$. Clearly it must be $u_1 \neq u_2$. If $u_1$ and $u_2$ do not lie on the same path of $T$, then $w$ must be the lowest common ancestor of $u_1$ and $u_2$, or one of its ancestors (excluding the root). Let $W$ be the label of the path from $w$ to the lowest common ancestor of $u_1$ and $u_2$: clearly $\ell(u_1)=X_1 \cdot a \cdot W \cdot \ell(w)$ and $\ell(u_2)=X_2 \cdot b \cdot W \cdot \ell(w)$, where $a$ and $b$ are distinct characters and $X_1$ and $X_2$ are strings of the same length, but at the same time it must be $X_1 \cdot a \cdot W = X_2 \cdot b \cdot W$, a contradiction. Assume thus that $u_1$ and $u_2$ lie on the same path of $T$: without loss of generality, let $u_1$ be an ancestor of $u_2$. 
Further assume that there is an edge $e=(u_1,v) \in E$ with $\ell(e)=a$. Then it must be that $\ell(u_1)=Y_1\cdot \ell(w)=X_1\cdot b\cdot Y_1$ and $\ell(u_2)=Y_2\cdot a \cdot Y_1\cdot \ell(w)=X_2\cdot Y_2\cdot a\cdot Y_1$ and $a \neq b$, but at the same time it must be that $a \cdot Y_1 = b \cdot Y_1$, a contradiction.
{$\Box$\vskip1ex}\end{proof}

\begin{lemma} \label{lemma:linear1}
Let $T \in [1..\sigma]^{n}$ be a string. There is an algorithm that computes $bord(W)$ for all right-maximal substrings $W$ of $T$ in $O(n)$ time and words of space.
\end{lemma}
\begin{proof}
We proceed as in Theorem \ref{theorem:apostolico}, but at every node $v$ of $\ST_T$ with $\ell(v)=W$, we store set $\mathcal{B}^{r}(W)$, sorted in lexicographic order, rather than $\mathtt{right}_{W}$. Recall that $\mathcal{B}^{r}(W)$ is the set of pairs $\{ (a,a|W) : a \in [1..\sigma], a|W \neq 0 \}$, i.e. a representation of the return arcs of Lemma \ref{lemma:returnArcs}. At every node $v$ we merge $\mathcal{B}^{r}(W)$ with the lexicographically sorted list of characters in $\Sigma_{T}^{\ell}(W)$, setting $bord(aW)=a|W$ (which might be zero) for every left extension $aW$. Once we know $bord(aW)$, we build $\mathcal{B}^{r}(aW)$ by copying the entire $\mathcal{B}^{r}(V)$, where $V$ is the longest border of $aW$, and by updating or inserting pair $(d,d|aW)$ using a linear scan of $\mathcal{B}^{r}(V)$, where $d$ is the character that precedes the suffix of $aW$ of length $bord(aW)$. This process touches every return arc of Lemma \ref{lemma:returnArcs} a constant number of times. 
{$\Box$\vskip1ex}\end{proof}

Lemma \ref{lemma:linear1} can be clearly applied to any trie $\mathcal{T}$ of size $n$ on an alphabet of size $\sigma$, in which every node stores its children in lexicographic order: the space used by such algorithm in addition to the trie is $O(\min\{n,\lambda\sigma\})$, where $\lambda$ is the length of a longest path in $\mathcal{T}$. However, Lemma \ref{lemma:returnArcs} does not generalize to radix trees, thus we cannot use it to bound the construction time of $\mathtt{left}$ arrays in Theorem \ref{theorem:apostolico}. To achieve this, it suffices to replace $\mathtt{left}$ arrays with suitable stacks:

\begin{lemma} \label{lemma:linear2}
Let $T \in [1..\sigma]^{n}$ be a string. There is an algorithm that computes $bord(W)$ for all right-maximal substrings $W$ of $T$, and for all substrings $W=Vb$ of $T$ where $b \in [0..\sigma]$ and $V$ is right-maximal, in $O(n)$ time and words of space.
\end{lemma}
\begin{proof}
We run the algorithm in Lemma \ref{lemma:linear1}, keeping $\sigma$ stacks $S_1,S_2,\dots,S_{\sigma}$. Assume that, when we visit the node $v$ of $\ST_T$ with $\ell(v)=W$, the top of stack $S_b$ for all $b \in \Sigma_{T}^{r}(V)$ stores value $W|b$ (which might be zero). Assume that, in the depth-first traversal of the suffix-link tree of $T$, we choose to visit the node $w$ with $\ell(w)=aW$ next: then, we iterate over all characters $b \in \Sigma_{T}^{r}(aW)$, we compute $aW|b$ by accessing position $bord(aW)$ of stack $S_b$, and we push $aW|b$ on $S_b$. This works since, if $aW$ can be extended to the right by character $b$, then every suffix of $aW$ can be extended to the right with character $b$ as well. This process takes overall $O(n)$ time, and the size of all stacks $S_1,\dots,S_{\sigma}$ is $O(\min\{n,\lambda\sigma\})$, where $\lambda$ is the length of a longest repeat of $T$, since every element in every stack can be charged to an edge of $\ST_T$.
{$\Box$\vskip1ex}\end{proof}

The information stored by Lemma \ref{lemma:linear2} is enough to compute in $O(n)$ time and space the longest border of all minimal rare words that occur at least once in $T$. Adapting Lemma \ref{lemma:linear2} to compute the longest border of all minimal \emph{absent} words of $T$ is also easy:

\begin{lemma} \label{lemma:maw}
Let $T \in [1..\sigma]^{n}$ be a string. There is an algorithm that computes $bord(W)$ for all minimal absent words $W$ of $T$ in $O(n+\mathtt{occ})$ time and words of space, where $\mathtt{occ}$ is the size of the output.
\end{lemma}
\begin{proof}
We proceed as in Lemma \ref{lemma:linear2}, but when we visit the node $v$ of $\ST_T$ with $\ell(v)=aW$ for some $a \in [1..\sigma]$, we assume that the top of stack $S_b$ for all $b \in \Sigma_{T}^{r}(W)$ stores value $aW|b$ (which might be zero). Assume that, in the depth-first traversal of the suffix-link tree of $T$, we choose to visit the node $w$ with $\ell(w)=caW$ next: then, we iterate over all characters $b \in \Sigma_{T}^{r}(aW)$, we compute $caW|b$ by accessing position $bord(caW)$ of stack $S_b$, and we push $caW|b$ on $S_b$. This works since, if $aW$ can be extended to the right by character $b$, then the proper suffix of the longest border of $caW$ can be extended to the right with character $b$ as well. Every element in every stack can be charged either to an edge of $\ST_T$ or to a minimal absent word of $T$, thus the size of all stacks is $O(\min\{n+\mathtt{occ},\lambda\sigma\})$, where $\lambda$ is the length of a longest repeat of $T$.
{$\Box$\vskip1ex}\end{proof}

\section{Detecting unusual words in small space}

As mentioned in the introduction, in order to compute the exact variance of a substring $W$ of a string $T$, we just need to compute functions $\phi(W)$ and $\gamma(W)$. Specifically, we just need to compute such functions on all right-maximal substrings of $T$, and on all substrings $Wa$ of $T$ such that $W$ is right-maximal \cite{apostolico2000efficient}. Since $\phi(W)$ and $\gamma(W)$ can be computed from $\phi(V)$ and $\gamma(V)$, where $V$ is the longest border of $W$, it is easy to see that we can adapt Lemma \ref{lemma:linear2} as follows. When we visit the node $v$ of $\ST_T$ with $\ell(v)=W$, we store $\phi(W)$ and $\gamma(W)$, and the top of stack $S_b$ for all $b \in \Sigma_{T}^{r}(W)$ stores the following values (which might be zero) in addition to $W|b$:
\begin{eqnarray*}
F_{b}(|W|) & = & D_{b}(|W|) \cdot \big( F_{b}(W|b) -2(|W|-W|b) \cdot G_{b}(W|b) +|T|-2|W|+W|b \big) \\
G_{b}(|W|) & = & D_{b}(|W|) \cdot \big( 1+G_{b}(W|b) \big)
\end{eqnarray*}

where $D_{b}(|W|) = \pi\big( W[W|b+1..|W|-1] \big) \cdot \mathbb{P}[b]$. Note that such values are computed recursively. We derive $\pi\big( W[W|b+1..|W|-1] \big)$ by keeping an additional stack that stores $\pi(V)$ for every suffix $V$ of $W$. The entire process can be implemented using just a data structure that supports $\mathtt{rangeDistinct}$ queries on $\BWT_T$, by implementing Lemma \ref{lemma:linear2} on top of the iterator described in Theorem \ref{thm:enumerator}:

\begin{theorem} \label{theorem:surprisingStrings}
Let $T \in [1..\sigma]^{n-1}\#$ be a string. Given a data structure that supports $\mathtt{rangeDistinct}$ queries on $\BWT_T$, we can compute $\phi(W)$ and $\gamma(W)$ for all right-maximal substrings $W$ of $T$, and for all substrings $W=Vb$ of $T$ such that $b \in [1..\sigma]$ and $V$ is right-maximal, in $O(nd)$ time and in $o(n+\lambda\sqrt{n})$ bits of space in addition to the input and the output, where $d$ is the time taken by the $\mathtt{rangeDistinct}$ operation per element in its output and $\lambda$ is the length of a longest repeat of $T$.
\end{theorem}
\begin{proof}
Recall that the iterator of Theorem \ref{thm:enumerator} returns, for every right-maximal substring $W$ of $T$, the set of all its right extensions in lexicographic order, but the set of all its left extensions \emph{in arbitrary order}. To implement Lemma \ref{lemma:linear2} in this case, we just need a temporary array $\mathtt{buffer}[1..\sigma]$ of $\sigma\log{n}$ bits that is initialized to all zeros at the beginning of the traversal. When the iterator visits substring $W$, we scan $\mathcal{B}^{r}(W)$ and we set $\mathtt{buffer}[a]=a|W$ for all $(a,a|W) \in \mathcal{B}^{r}(W)$. Then, for every left extension $a_i$ of $W$ provided by the iterator, we compute $bord(a_{i}W)$ by accessing $\mathtt{buffer}[a_i]$, and we proceed as in Lemma \ref{lemma:linear2}. Once we have finished processing substring $W$, we reset $\mathtt{buffer}$ to its previous state by setting $\mathtt{buffer}[a]=0$ for all $(a,a|W) \in \mathcal{B}^{r}(W)$. The claimed space complexity comes from Theorem \ref{thm:enumerator} and from our assumption that $\sigma \in o(\sqrt{n}/\log{n})$.
{$\Box$\vskip1ex}
\end{proof}

For statistical reasons only substrings of length $O(\log_{\sigma}{n})$ are candidates for being over- or under-represented \cite{apostolico2003monotony}, so the space complexity of Theorem \ref{theorem:surprisingStrings} is effectively $o(n)$, every surprising substring can be encoded in a constant number of machine words, and thus it can be printed in constant time\footnote{We assume the word RAM model of computation with words of size $\Omega(\log{n})$ bits, in which all standard operations including multiplication have unit cost.}. The technique described in Theorem \ref{theorem:surprisingStrings} can be used to apply Lemma \ref{lemma:linear1} to tries whose nodes do not store the list of their children in lexicographic order. Moreover, it is easy to adapt Theorem \ref{theorem:surprisingStrings} to \emph{output} all candidate over- and under-represented substrings whose statistical score matches a user-specified criterion, by keeping an additional stack of characters of size $\lambda\log{\sigma}$ and by exploiting the fact that the iterator in Theorem \ref{thm:enumerator} returns the length of every substring it visits.

Reducing the number of patterns displayed by a data mining algorithm is key for making it useful in practice. According to the monotonicity of the scores described in Section \ref{sec:intro}, we can limit the search for over-represented (respectively, under-represented) substrings to maximal repeats (respectively, to minimal rare words): this observation was already implicit in \cite{apostolico2003monotony}, and Theorem \ref{theorem:surprisingStrings} can be easily adapted to consider only such candidates. Moreover, in a practical implementation we can compute and store $\phi(W)$ and $\gamma(W)$ just for maximal repeats, since the longest border of a maximal repeat is itself a maximal repeat. We can also avoid storing numbers $F_b(|W|)$, $G_b(|W|)$ and $W|b$ for all $b \in \Sigma_{T}^{r}(W)$ on the stacks of Lemma \ref{lemma:linear2}, whenever $W$ is a right-maximal substring of $T$ that cannot be written as $aV$ for a character $a$ and a maximal repeat $V$. Within the working space budget of Theorem \ref{theorem:surprisingStrings}, but in time $O(nd+\mathtt{occ})$, we can also compute the border and the statistical scores of all $\mathtt{occ}$ minimal absent words of $T$, by adapting Lemma \ref{lemma:maw} to work on Theorem \ref{thm:enumerator}: such strings are the only strings that do not occur in $T$ which could be under-represented in $T$, however they were not reported in previous works \cite{apostolico2003monotony,apostolico2000efficient}. The ability to assign a statistical score to minimal absent words could be useful also in other contexts, for example in choosing which minimal absent words should be displayed to the user, since their total number is $\Theta(n\sigma)$ in the worst case.

\section{Implementation and experiments}

The algorithms described in this paper are practical: a prototype implementation that scores all maximal repeats and minimal rare words of a DNA string of length $14.8 \cdot 10^6$ uses on average just $12 \cdot 10^3$ bits for the stack (which can fit in the L1 cache of current processors), with a few sudden bursts that reach up to approximately $6 \cdot 10^5$ bits (which can fit in the L2 cache of current processors): see Figure \ref{fig:stats} (left). Since borders are short on average (constant for a random string), border information takes a negligible fraction of the stack, which is approximately equally divided between BWT intervals and quantities related to the variance. Since most borders are short, most recursive accesses of the algorithm are directed to the bottom of the stack (Figure \ref{fig:stats}, right), so the hardware can automatically move portions of this region to the L1 cache. Moreover, if the user knows that the length of the longest border of every substring of the string under analysis is upper-bounded by a constant $\beta$, the stack can be made even smaller by pushing border and variance information just for string of length at most $\beta$. 

Our prototype implementation compares favourably to the previous implementation of Theorem \ref{theorem:apostolico} described in \cite{apostolico2004verbumculus}: given the BWT represented as a wavelet tree, our implementation scores \emph{all} maximal repeats and minimal rare substrings in approximately 33 seconds on a 2.50 GHz Intel Xeon E5-2640, while the previous implementation takes approximately 57 seconds and a peak of 6 gigabytes to build a truncated suffix tree and to score only candidate substrings of length at most 12, and 2 (respectively, 4) minutes and a peak of 14 gigabytes to build a truncated suffix tree and to score only candidate substrings of length at most 24 (respectively, 36). Being essentially a tree traversal, our algorithm is intrinsically parallel: we are in the process of developing a shared-memory, multithreaded implementation, and of testing the speedup induced by using more than one core.

\begin{figure}[t]
\centering
\includegraphics[scale=0.3]{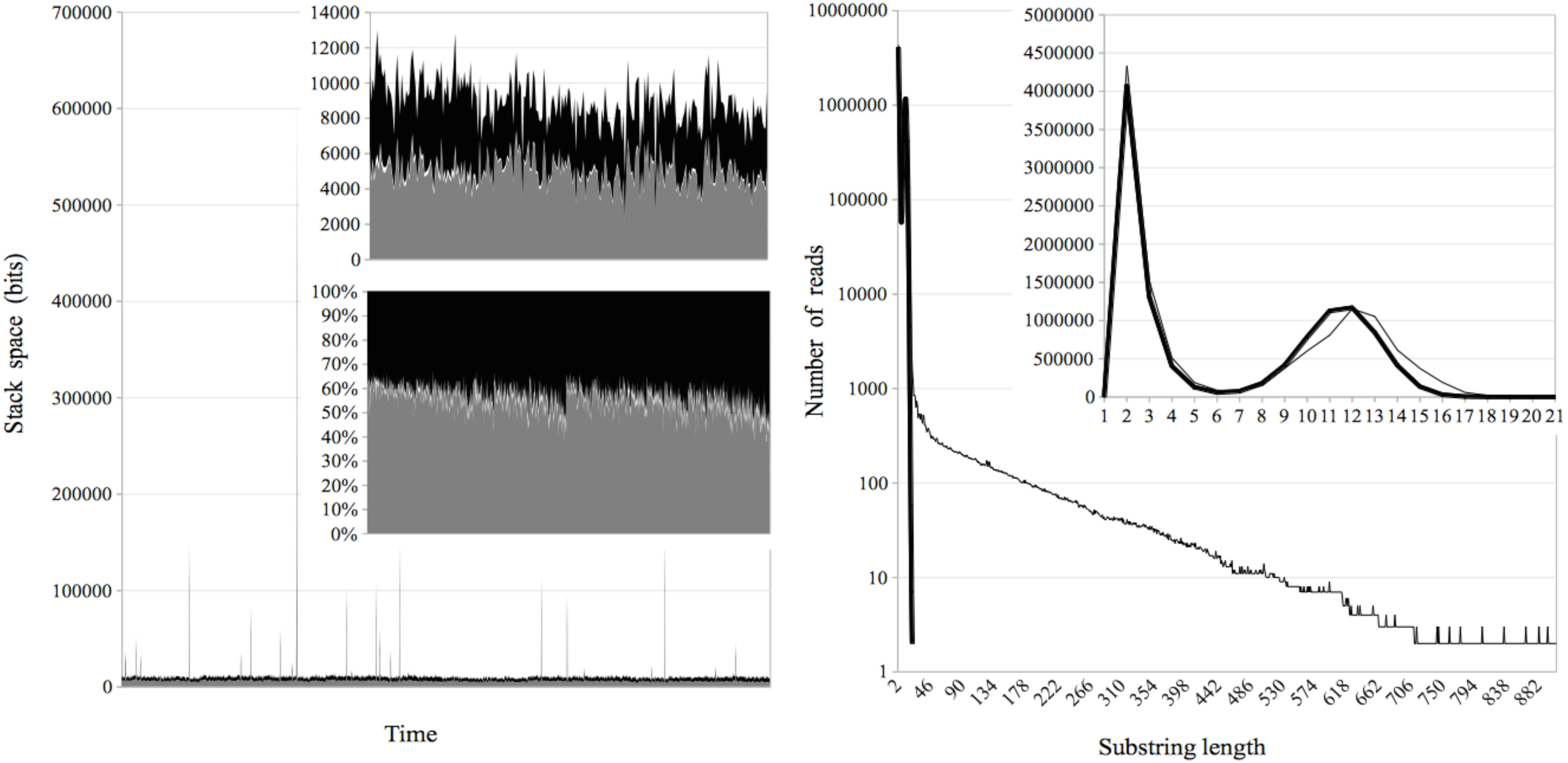}
\caption{A prototype implementation that scores all maximal repeats and minimal rare words of the genome of \emph{Sorangium cellulosum} ($\mathtt{NC\_021658}$ in NCBI). Left panel: number of bits taken by the stack in an entire run, sampled every two second. Top insert: detail of a typical subinterval without spikes. Bottom insert: percent composition of the stack in an entire run. Light gray: information to traverse the suffix-link tree using the BWT. White: border information. Black: variance information. Right panel: number of recursive read requests per string length. Thin line: $\mathtt{NC\_021658}$, thick line: a random reshuffle of $\mathtt{NC\_021658}$. The insert details the interval of lengths $[1..21]$. The peak around length 12 is caused by the retrievals of $\delta(W)$ generated by substrings of expected length.}
\label{fig:stats}
\end{figure}

\section{Acknowledgements}

We thank Alberto Apostolico for illuminating the details of the algorithms in \cite{apostolico2000efficient}, and Stefano Lonardi for providing the source code of the implementation described in \cite{apostolico2004verbumculus}.

\bibliographystyle{plain}
\bibliography{surprisingStrings}

\end{document}